\documentclass{amsart}

\usepackage{amsmath,amsthm,verbatim,xcolor}

\usepackage[margin=1in]{geometry}

\newcommand{\off}{\mathrm{OFF}}
\newcommand{\on}{\mathrm{ON}}
\newcommand{\rep}[1]{REPEAT(#1)}
\newcommand{\flip}[2]{FLIP(#1,#2)}
\newcommand{\osc}[3]{OSCILLATE(#1,#2)}
\newcommand{\seeC}[1]{SEE(#1)}
\newcommand{\dec}{DECLARE}
\newcommand{\done}{\mathrm{DONE}}
\newcommand{\upc}{\mathrm{UP}}
\newcommand{\nextc}{\mathrm{NEXT}}
\newcommand{\ready}{\mathrm{READY}}
\newcommand{\ceil}[1]{\left\lceil #1 \right\rceil}
\newcommand{\sch}{\Sigma}

\newtheorem{thm}{Theorem}[section]

\newtheorem{lemma}[thm]{Lemma}

\newtheorem{theorem}[thm]{Theorem}
\newtheorem{remark}[thm]{Remark}

\newcommand{\term}{\emph}
\newcommand{\tit}[1]{\textsc{#1}}
\newcommand{\com}{\ \ \ \ $\backslash\backslash$}

\title{Prisoners, Rooms, and Lightswitches}

\author{Daniel M. Kane}
\address{Department of Mathematics / Department of Computer Science and Engineering\newline\indent University of California, San Diego\newline\indent 9500 Gilman Drive \newline\indent La Jolla, CA 92093}
\email{dakane@ucsd.edu, aladkeenin@gmail.com}

\author{Scott Duke Kominers}
\address{Harvard Business School and Department of Economics\newline\indent Harvard University\newline \indent Rock Center 219,
Harvard Business School\newline\indent
Soldiers Field, Boston, MA 02163 }
\email{kominers@fas.harvard.edu, skominers@gmail.com}

\thanks{Kane gratefully acknowledges the support of an NSF Postdoctoral Fellowship, NSF CAREER Award 1553288, and a Sloan Research Fellowship. Kominers gratefully acknowledges the support of a Harvard Mathematics Department Highbridge Fellowship, an NSF Graduate Research Fellowship, National Science Foundation Grants CCF-1216095 and SES-1459912, an AMS-Simons Travel Grant, the Harvard Milton Fund, and the Ng Fund and the Mathematics in Economics Research Fund of the Harvard Center of Mathematical Sciences and Applications.}

\begin{document}

\begin{abstract}
We examine a new variant of the classic prisoners and lightswitches puzzle: A warden leads his $n$ prisoners in and out of $r$ rooms, one at a time, in some order, with each prisoner eventually visiting every room an arbitrarily large number of times. The rooms are indistinguishable, except that each one has $s$ lightswitches; the prisoners win their freedom if at some point a prisoner can correctly declare that each prisoner has been in every room at least once. \textit{What is the minimum number of switches per room, $s$, such that the prisoners can manage this?} We show that if the prisoners do not know the switches' starting configuration, then they have no chance of escape---but if the prisoners do know the starting configuration, then the minimum sufficient $s$ is surprisingly small. The analysis gives rise to a number of puzzling open questions, as well.
\end{abstract}

\maketitle

\section{Introduction}

The following puzzle is well-known: \begin{quote}There are $n$ prisoners in a prison.  The warden offers a deal:  He will lead the prisoners into a particular room one at a time in some order, with the  guarantee that each prisoner will eventually be led into the room arbitrarily many times.  At any point, a prisoner may declare that all the prisoners have been in the room.  If the declaring prisoner is correct, then the prisoners are freed.  Otherwise, they are executed(!).

The prisoners are allowed to confer ahead of time to agree upon a strategy, but are allowed no direct communication after the exercise starts.  The room that they are led into is completely featureless except for a lightswitch, which starts in the $\off$ position.  What lightswitch flipping strategy guarantees the prisoners' freedom?
\end{quote}

This ``One-Bulb Room'' problem appears in Winkler's \textit{Mathematical Puzzles}~\cite[p.~103]{Winkler};  Winkler remarks that it is also appeared in \textit{The Emissary}~\cite{Emiss} and as a puzzler on \textit{Car Talk}~\cite{NPR}. Dehaye, Ford, and Segerman~\cite{dehaye2003one} have studied a similar problem, in which the prisoners may synchronize their actions with a global clock.

Despite its popularity, however, the One-Bulb Room problem appears to have seen little generalization.  Indeed, the second author, Kominers, and Chen~\cite{HCMR} posed a generalization, inquiring about what happens when the number of rooms is increased to $r>1$.\footnote{The problem was later featured in one of the second author's \textit{Bloomberg Opinion} puzzle columns~\cite{BOP}; the solution presented there and in \cite{HCMR2} corresponds to the protocol we present in Section~\ref{sec:2s}, although our analysis here is far more formal.}  This query, in turn, has a number of variations. Some turn out to be surprisingly subtle---to whit the original solution of~\cite{HCMR} contained an error, which was spotted by the first author.  Discussions about a corrected version of the problem (\cite{HCMR2}) have led to this article, which rigorously investigates the circumstances under which the prisoners may win their freedom (and those in which they are doomed to failure).

\subsection{A Solution for $n$ Prisoners, One Room}\label{oneroomsec}

With only one room to track, the prisoners have a fairly simple escape strategy.  They select a \term{leader}, who will keep count of the number of prisoners who have entered the room.  The other prisoners signal that they have been in the room by turning the lightswitch $\on$ the first time they are able to do so; the leader acknowledges these signals by turning the switch $\off$ again.  Given that each prisoner only signals once, the leader will know that all the prisoners have been in the room once he has acknowledged $n-1$ signals.

Of course, if $n>1$, the leader has entered the room at least once by the time he has acknowledged $n-1$ signals; at that time he can declare immediately.  (When $n=1$, of course, the leader must wait until he has been in the room before declaring.)

\subsection{A Formal Framework}

Although the $n$ prisoner, one room solution just discussed is fairly straightforward, much of our later discussion will be far more complex.  Therefore, we introduce a notation for prisoners' solution protocols at the outset, using the $n$ prisoner, one room solution as an example.

We say that a room with $s$ switches is \emph{in configuration $(a_1,\ldots, a_s)$} if its switches display the values $a_1,\ldots, a_s$ (in sequence).  For example, a room with one switch has two possible configurations: $(\on)$ and $(\off)$. In general, it will be somewhat cumbersome to describe configurations as lists of switch values, so we often give names to configurations instead.

In the solution described above, all the prisoners who are not the leader follow a very simple algorithm: they wait until they see the lightswitch $\off$, and then turn it $\on$.  We notate this procedure as ``\flip{$(\off)$}{$(\on)$}.''  In general, we define \begin{quote}\flip{$A$}{$B$}: Wait until you see a room in configuration $A$ and then reconfigure it to $B$.\end{quote}

The leader, meanwhile, must apply a more complicated algorithm.  First, assuming $n>1$, the leader must turn the lights off $n-1$ times.  We could express this by writing ``\flip{$(\on)$}{$(\off)$}'' $n-1$ times, but this seems cumbersome. We instead write:
\begin{quote}
\rep{$n-1$}\newline
\indent \flip{$(\on)$}{$(\off)$},
\end{quote} where \rep{$k$} indicates that the prisoner should repeat the nested actions $k$ times. After the counting phase is completed, the leader must \term{declare}, announcing that everyone has entered the room; this operation is written ``\dec.''  The leader's complete algorithm in the solution to the one-room case is therefore: \begin{quote}
\rep{$n-1$}\newline
\indent \flip{$(\on)$}{$(\off)$}\newline
\dec.
\end{quote}

As we have already observed, if $n=1$, then the leader (who is the only player by default) must follow a slightly different algorithm.  He must wait until he enters a room, and then must \dec.  Equivalently, since the switch starts $\off$, he must wait until he sees the configuration $(\off)$.  We define the operation 
\seeC{$a$}, hich means that a prisoner waits (i.e., he progresses no further through his algorithm) until he enters a room that has configuration $a$.  The full solution to the one-room problem when $n=1$ is therefore:
\begin{quote}
\seeC{$(\off)$}\newline
\dec.
\end{quote}
(Although the \seeC{$a$} operation is equivalent to the ``trivial'' flipping operation \flip{$a$}{$a$}, it is useful to distinguish these two operations for clarity.)

As the solutions we discuss shall often require trivial modifications in the case $n=1$ (as occurs in the $n$ prisoner, one room problem), we will hereafter assume $n>1$ except where otherwise noted.

In order to add clarity to these protocols we will add comments at the end of some lines delineated by double backslashes:
\begin{quote}
\flip{$0$}{$1$} \com{ This is a comment.}
\end{quote}

\subsection{A Note on Starting Configuration}

Note that protocol we have just described assumes that the room is known to start with its switch in the $\off$ state.  If the room is known to start in the $\on$ state, an analogous protocol may be applied.  On the other hand, if the room begins in an unknown state, a slightly more complicated approach must be used.  In particular, the prisoners may use the following protocol:

\begin{quote}
\tit{Leader's Algorithm:}\newline
\rep{$2n-2$}\newline
\indent \flip{$(\on)$}{$(\off)$}\newline
\dec;
\end{quote}

\begin{quote}
\tit{Other Prisoners' Algorithm:}\newline
\rep{$2$}\newline
\indent \flip{$(\off)$}{$(\on)$}.
\end{quote}
The analysis of this protocol is similar to that of the simpler one for a known starting configuration.  The differences here are that each non-leader signals twice, and that, if the room stats in the $(\on)$ state, the leader will ``acknowledge'' an extra signal.  This causes him to declare before all other prisoners  have signalled twice---instead, he declares after all prisoners but one have signalled twice, and the remaining prisoner has signalled once.  Nonetheless,  whenever the leader declares, all prisoners will have entered the room at least once.

\subsection{$n$ Prisoners, $r$ Rooms}

Having thus handled the problem as stated, we consider generalizations in which the prison has $r\geq 1$ rooms that the prisoners might be led into.  There are several slight variants of this generalization; we discuss them in ascending order of difficulty.

\subsubsection{Distinguishable Rooms}

If the different rooms are disguisable, then the prisoners can treat each room as a separate, parallel instance of the original problem.  More generally, if the rooms may be partitioned into classes of mutually indistinguishable rooms, then each class may be addressed separately (at least under the assumption that the same prisoner would have declared in each sub-instance).  Hence, to keep the problem interesting, we only examine the case in which the only features distinguishing  the rooms are the configurations of their lightswitches.

\subsubsection{Each Prisoner Visits at least One Room}

With multiple rooms, the warden might relax his requirements of the prisoners.  In particular, he might ask that the prisoners declare only after each has visited at least one room.  However, this problem can be solved using an algorithm similar to that used in the one-room case.  In particular, even for an unknown starting configuration, the following protocol wins the prisoners freedom:
\begin{quote}
\tit{Leader's Algorithm:}\newline
\rep{$(r+1)(n-1)$}\newline
\indent \flip{$(\on)$}{$(\off)$}\newline
\dec;
\end{quote}
\begin{quote}
\tit{Other Prisoners' Algorithm:}\newline
\rep{$r+1$}\newline
\indent \flip{$(\off)$}{$(\on)$}.
\end{quote}

\subsubsection{Each Prisoner Visits All Rooms}\label{validSec}

From the preceding analysis, we are narrowed to a case in which the warden can really make trouble for the prisoners:  In this case, we have $r$ rooms, distinguishable from each other only by the states of their lightswitches, and the warden requires that each prisoner must have visited every room before some prisoner declares.  In the interest of fairness, the warden must grant the prisoners the guarantee that, if they wait long enough, each prisoner will eventually be led into every room an arbitrarily large number of times; we say that a schedule of room visits is \term{valid} if it has this property. We are now left with the following question:
\begin{quote}
What is the minimum number of switches per room, $s$, so that the prisoners have a protocol ensuring that they can win their freedom under any valid schedule of room visits?
\end{quote}

\section{A Negative Result when the Starting Configuration is Unknown}
We begin by supposing that the rooms' starting configurations are unknown.  Unlike the one-room case, in which this difficulty can be circumvented with only a slight modification of the prisoners' algorithm, if there are $r>1$ rooms and their initial configurations are unknown, then \textit{the prisoners have no protocol that is guaranteed to work}. 

In order to prove this kind of impossibility result, it will be important to describe the adversarial strategy for the warden. We begin with the following Lemma:

\begin{lemma}\label{lem:s12}
Suppose that each room has a finite number, $s$, of switches. Fix a deterministic strategy for the prisoners. For that strategy, there is a starting configuration for the rooms and a pair of schedules, $\sch_1$ and $\sch_2$, having the following properties:
\begin{enumerate}
	\item\label{oneRoomProperty} Under $\sch_1$, the prisoners will only ever visit one of the rooms.
	\item Under $\sch_2$, each prisoner will visit each room infinitely often, i.e. the schedule is valid in the sense described in Section~\ref{validSec}.
	\item The schedules $\sch_1$ and  $\sch_2$ are indistinguishable from the prisoners' perspectives. In particular, if the prisoners execute their strategy for $\sch_1$ or execute their strategy for $\sch_2$ each prisoner will see the same sequence of room configurations in either case.
\end{enumerate}
\end{lemma}
This will be enough since if the warden leads the prisoners through these schedules, as they are indistinguishable, to the prisoners the prisoners must either eventually declare on both or never declare on both. But this would cause them to either declare incorrectly in $\sch_1$ or never declare on the valid schedule $\sch_2$.
\begin{proof}
In order to produce these schedules $\sch_1$ and $\sch_2$, we start with some given room configuration, say $C$.  Consider a room in configuration $C$. Fix an ordering of the prisoners, and consider sending them into that room repeatedly in that order. After each pass through such a cycle, record the current configuration of the room in question. Since there are finitely many configurations, some configuration, $D$, must show up infinitely often.

We/the warden start/s with the rooms configured so that one room is in configuration $C$ and the rest are in configuration $D$. In schedule $\sch_1$, we send all the prisoners through the first room in the specified order repeatedly. It is clear that this satisfies our Property~\eqref{oneRoomProperty}.

To construct schedule $\sch_2$, we maintain an ordered list of the rooms and a  \textit{current room} indicator, initially set to indicate the first room. The warden repeatedly sends the prisoners in order into current room; however when at the end of any cycle if the current room is in configuration $D$, then he switches the current room indicator to the next room on the ordered list.

To prove that schedules $\sch_1$ and $\sch_2$ are indistinguishable, we note first that in both cases the order in which prisoners are chosen to enter rooms is the same. Furthermore, we claim that for each $k$, the state of the current room in step $k$ of $\sch_2$ is the same as the state of the first room at step $k$ of $\sch_1$. We prove this by induction on $k$. For $k=1$, the current room is in state $C$ in either case. If the current rooms were in the same state at each time leading up to $k$, then they will be the same at step $k$ because the $k$-{th} prisoner will in either schedule have the same history and will visit a room in the same state, and thus will make the same change to it. The one slight twist that needs to be added is that if we have completed a pass through the prisoners and the current room is in state $D$, the warden will then change the current room indicator in schedule $\sch_2$. However, this will necessarily change the current room from one room in state $D$ to another, and will not affect our claim.

We have left to prove that $\sch_2$ is valid. For this, we note that in $\sch_1$, by assumption, is it the case infinitely often that at the end of a {cycle through the prisoners}, that the warden finds the first room in configuration $D$. Therefore, by the indistinguishability result already proven, in schedule $\sch_2$ the warden will infinity often at the end of a cycle find the current room in configuration $D$. Therefore, in $\sch_2$, the current room indicator will change infinitely often. However, each time the current room indicator changes, each prisoner will visit the new current room at least once before the indicator changes again. Since the indicator changes infinitely often, and since the warden will change it in a sequence that cycles through all of the rooms infinitely often, each room will become the current room infinitely many times. Thus, the schedule will send each prisoner to each room infinitely many times.
\end{proof}

Given Lemma~\ref{lem:s12}, it is not hard to show our impossibility result. The idea is that the warden will send prisoners through one of the two  schedules constructed in the lemma, and either the prisoners will declare incorrectly in $\sch_1$ or fail to declare in $\sch_2$.
\begin{thm}\label{unknownStartThrm}
Assume $r>1$ rooms and that the number of switches per room, $s$, is finite.  Then for any deterministic strategy for the prisoners, there is an  set of initial room configurations and an associated valid schedule $\sch$ so that if the prisoners visit rooms according to $\sch$, they will either declare incorrectly or fail to declare.
\end{thm}
\begin{proof}
We use the initial room configurations specified by Lemma~\ref{lem:s12}. In order to find the schedule, we first consider what happens if the prisoners are sent into rooms according to schedule $\sch_2$. If they do not declare, we have a valid schedule for which the prisoners never declare, and we are done. If they do declare, then we note that by the indistinguishability property that the prisoners must also declare in schedule $\sch_1$. This declaration will necessarily be incorrect, as $\sch_1$ only involves visits to a single room (and $r>1$). However, $\sch_1$ is not a valid schedule. We construct a valid schedule $\sch_1$ by noting that the declaration under $\sch_1$ will occur after some finite number $\ell$ of visits. We thus pick a schedule $\sch_1'$ that agrees with $\sch_1$ for the first $\ell$ visits and after that sends all prisoners to all rooms infinitely many times in whatever order is desired. It is now clear that $\sch_1'$ is valid, however since it agrees with $\sch_1$ on the first $v$ steps, the prisoners will still declare on step $v$---before they have all visited all the rooms.
\end{proof}

\section{A Solution for When the Starting Configuration Is Known}

Given the impossibility result presented in the previous section, we henceforth focus on the case in which the rooms' starting configurations are known in advance.

\subsection{Arbitrary Starting Configuration}

We start with the general case, in which the starting configurations, while known, may be arbitrary.  Arbitrary starting configurations have the potential to make the prisoners' task difficult, since seeing a room in a given configuration could just mean that the room started in that configuration.

The prisoners would prefer to work in a simple, ``canonical'' starting configuration, for example the one in which where all of the switches start in the $\off$ position.  Fortunately, there is an approach that allows one to use a strategy that works for the all $\off$ starting configuration (satisfying some mild extra conditions) to produce a strategy that works for an arbitrary, known starting configurations.

\begin{lemma}\label{reductionlemma}
Suppose that, given $n$, $r$, and $s$, the prisoners have a winning protocol if all of the switches start in the $\off$ position. Suppose additionally that
\begin{enumerate}
\item the winning protocol  makes use of two room configurations, here denoted $0$ and $1$, where $0$ is the all-$\off$ configuration (and $1$ is some other configuration); and
\item   one of the prisoners is designated as the \emph{leader}, and all non-leader prisoners will ignore all rooms they are sent to until they see a room in a configuration other than $0$ or $1$.
\end{enumerate}
Then, the prisoners have a winning protocol for any (known) set of starting configurations.
\end{lemma}
\begin{proof}
The idea of this protocol is to run the old winning protocol preceded by a protocol that puts all switches in the $\off$ position.  We suppose that of the $r$ rooms, it is known that $r_0$ start in configuration $0$ and $r_1$ start in configuration $1$.  In fact, we will not need to know the multiplicities of the other configurations.

The protocol starts with the leader changing the configurations of the $r-r_0-r_1$ rooms not in configuration $0$ or $1$ to configuration $1$.  We next need a way of informing the other prisoners that these rooms have been cleared out.  This is done by changing rooms between configurations $1$ and $0$. In particular, each non-leader will attempt to change rooms from configuration $0$ to configuration $1$ a total of $r_0+1$ times before starting on his old protocol.  This number is chosen so that the prisoner cannot possibly see that many rooms in configuration $0$ without someone changing rooms to configuration $0$.  In the meantime, the leader will change rooms from configuration $1$ to configuration $0$.  He will do this $r-r_0+(n-1)\cdot (r_0+1)$ times.  This ensures that each other prisoner has changed $r-r_0$ rooms from $0$ to $1$ and that all rooms are now in configuration $0$.  After this point we are ready to begin the old protocol.  To summarize, here are the prisoners' strategies---where we use a new directive \flip{$*$}{$1$} meaning ``wait until you see a room not in configuration $0$ or $1$ and change that room to configuration $1$.''
\begin{quote}
\tit{Leader's Protocol:}\newline
\rep{$r-r_0-r_1$}\newline
\indent \flip{$*$}{$1$}\newline
\rep{$r-r_0+(n-1)\cdot (r_0+1)$}\newline
\indent \flip{1}{0}\newline
RUN OLD PROTOCOL
\end{quote}

\begin{quote}
\tit{Other Prisoners' Protocol:}\newline
\rep{$r_0+1$}\newline
\indent \flip{$0$}{$1$}\newline
RUN OLD PROTOCOL
\end{quote}

In order to show that the preceding protocol works, we will need to verify that:
\begin{enumerate}
\item No prisoner begins to run their old protocol before the leader completes his first REPEAT loop.
\item Between the end of the leader's first REPEAT loop and when he begins to run his old protocol, all rooms are in configuration $0$ or $1$.
\item When the leader begins to run his old protocol, the other prisoners have all started to run their old protocols, but have ignored all room they have seen since they started doing so, and all rooms are in configuration $0$.
\item Eventually the leader will reach the RUN OLD PROTOCOL step.
\end{enumerate}
Once we have proven these statements we will be done, since statement (4) implies that eventually the leader begins to run his old protocol, statement (3) implies that at that time
\begin{itemize}
	\item all non-leader prisoners are acting as if they were at the start of their old protocol and \item all rooms are in state $0$. \end{itemize}Therefore, from that point in time, it is as if all prisoners were running the old protocol with the correct starting configuration. Since the warden must send each prisoner into each room arbitrarily many times from that  point (in order for the sequence of visits to be valid), the correctness of the old protocol implies the correctness of the new one.

Statement (1) holds because each time a prisoner FLIPs a $0$ to a $1$, the number of rooms in configuration $0$ decreases.  The only way this number can increase is either after the leader finishes his first REPEAT loop or after some other prisoner begins running their old protocol.  Since the number of starting $0$s is less than the number that must be changed, no non-leader can begin to run their old protocol until some room changes to configuration $0$.  Therefore, the first prisoner to begin to run their old protocol must have done so after the leader completed his first REPEAT loop.

For statement (2), we first note by (1) that until the leader completes the first REPEAT loop, no other prisoner begins their old protocol.  Therefore, until that time, the only way that the number of rooms not in state $0$ or $1$ changes is that the number {decreases by one} each time the leader executes his \flip{$*$}{$1$} command.  Therefore when the leader finishes his first REPEAT loop there are no such rooms remaining. Thus, at the end of the leader's first REPEAT loop, all rooms are in configuration $0$ or $1$. We note that between that time and the end of the leader's second REPEAT loop, the only way that a room can be put into a configuration other than $0$ or $1$ would be if another prisoner who has started to execute their old protocol does so. However, by assumption, non-leader prisoners who are executing their old protocols will not pay attention to any rooms (much less change their configurations) until they have seen one in a configuration other than $0$ or $1$. However, there is no way that a prisoner can be the first to do this as they will need to have first seen a room in a state other than $0$ or $1$ which must have been produced by some even earlier prisoner.

Statement (3) is proven by considering the number of rooms in configuration $0$.  This number {increases by one} when the leader runs a \flip{$1$}{$0$}, {decreases by one} when another prisoner runs a \flip{$0$}{$1$}.  Since statement (2) implies that none of the non-leaders have reconfigured rooms or executed any commands since the starting to run their old protocols, these are the only ways this number can change until the leader starts to run his old protocol.  The number of rooms in configuration $0$ starts at $r_0$.  In order for the leader to begin the old protocol, this number must increase $r-r_0+(n-1)\cdot (r_0+1)$ times. However since the number of rooms in configuration $0$ can never exceed $r$, this is only possible if it has decreased at least $(n-1)\cdot (r_0+1)$ times.  This many decreases can happen only if each of the other prisoners run their \flip{$0$}{$1$} the full $r_0+1$ times and begin running their old protocols.

Statement (4) is a liveness condition that can be proven by looking carefully at the analysis thus far.  First, we show that the leader will eventually finish his first REPEAT loop.  This is because he executes a \flip{$*$}{$1$} once whenever he enters any of the rooms that did not start as $0$ or $1$ for the first time.  Since there are $r-r_0-r_1$ of these, eventually he has visited all of them and completed the loop.  Next, we note that there will never be a time at which no prisoner can make progress on his REPEAT loop. Indeed, our analysis thus far shows that if all of the non-leaders have completed their REPEAT loops, there will be as many rooms in configuration $1$ as iterations  left in the leader's loop.  Therefore, if the leader enters the appropriate room, he will make progress through his protocol.  If both the leader and some non-leader have FLIPs to perform, then either there is a $1$ for the leader to FLIP to a $0$ or a $0$ for the non-leader to FLIP to a $1$.  As validity guarantees that every prisoner will be sent to every room as many times as we need, if the prisoners wait long enough, then eventually one of them will complete one of their FLIP commands---and this can only happen a bounded number of times before everyone starts to run their old protocols.
\end{proof}

\subsection{All Switches Start $\off$}

Given the reduction proven in Lemma \ref{reductionlemma} we henceforth focus most of our effort on the case in which all rooms start in a specific known configuration---in particular,  the case in which all switches start in the $\off$ position.

\subsubsection{Na\"ive Solutions}\label{naive solutions section}

Some simple solutions come to mind quickly if we allow potentially large values of $s$.  If we had infinitely many switches in each room, the prisoners could use them to encode (in English, translated by some means into binary) any information they want.  They could use this to record the complete history of each room---the sequence of visits by prisoners, the names of these prisoners, the rooms previously visited by these prisoners, the configurations of those previously visited rooms, and any proofs of the Riemann Hypothesis that they have discovered in the meantime.  Eventually the rooms will become distinguishable, based on the first time that some particular prisoner, say Alvin, visited a given room.  Once this has happened it is trivial for any prisoner to (eventually) verify from each room's history that every room has been visited by every prisoner.

To implement an analogous protocol with only finitely many switches, we note that it suffices to store in each room's configuration:
\begin{itemize}
\item[(a)] for each prisoner, $p$, whether or not $p$ has visited the room;
\item[(b)] a number distinguishing each room, such as the $k$ so that this was the $k$-{th} distinct room visited by Alvin.
\end{itemize}
We can encode (a) using $n$ switches per room (one for each prisoner); and can encode (b) with $r$ switches using a unary encoding.  Hence we only need $s=n+r$ in order for the prisoners to have a winning protocol.

On the other hand, using unary counters is somewhat inefficient.  The value of $k$ can be encoded in binary using $\ceil{\log_2(r+1)}$ switches.  Furthermore, if prisoners wait to indicate their visits to a room until after it has been assigned identifiers by Alvin, they only need to store the number of prisoners who have visited the room---and anyone can declare once he verifies that each room has been visited by all $n$ prisoners. Such a counter can be implemented easily with $\ceil{\log_2(n+1)}$ switches; hence, we only require $$s=\ceil{\log_2(r+1)}+\ceil{\log_2(n+1)}.$$

A slight optimization on the preceding protocol removes the need to identify the rooms, so long as the prisoners only track their presence in each room sequentially. In particular, if we identify $n+1$ distinguished configurations, denoted, $0,1,\ldots,n$, where $0$ is the all-$\off$ configuration, and assign each prisoner an identifier $i\in\{1,2,\ldots,n\}$, we can use the following protocol:

\begin{quote}
\tit{Prisoner $i$'s Algorithm ($i<n$):}\newline
\rep{$r$} \newline
\indent \flip{$i-1$}{$i$}
\end{quote}

\begin{quote}
\tit{Prisoner $n$'s Algorithm:}\newline
\rep{$r$} \newline
\indent \flip{$n-1$}{$n$}\newline
\dec.
\end{quote}
When the preceding protocol is followed, for a room to be in configuration $i$ it must have been visited sequentially by prisoners, $1,2,\ldots,$ and $i$.  The declaration will not be made until each room is in configuration $n$, in which case each room has been visited by every prisoner.  For there to be $n+1$ available configurations, there must be at least $\ceil{\log_2(n+1)}$ switches per room, and hence $$s=\ceil{\log_2(n+1)}$$ suffices.

\subsubsection{A Less Na\"ive Solution}The solutions just described store a large amount of data across the rooms; a lower-overhead solution attempts to run the original one-room protocol for each room sequentially.  A simple version of this protocol requires six distinct configurations which we will call $\off$ (the initial room configuration), $\done$, $0,1,0'$ and $1'$.  During the running of the protocol:
\begin{itemize}
\item Rooms in the  $\off$ configuration have not yet been modified.
\item Rooms in the $\done$ configuration have been visited by all prisoners and will not be modified again.
\item Configurations $0$ and $1$ are used to implement the one-room protocol discussed in Section~\ref{oneroomsec}.
\item Configurations $0'$ and $1'$ are used  to communicate to each prisoner that it is time to move on to the next room.
\end{itemize}
Formally we use:
\begin{quote}
\tit{Leader's Algorithm:}\newline
\rep{$r$}\newline
\indent \flip{$\off$}{$0$}\newline
\indent \rep{$n-1$}\newline
\indent\indent \flip{$1$}{$0$}\newline
\indent \flip{$0$}{$0'$}\newline
\indent \rep{$n-1$}\newline
\indent\indent \flip{$1'$}{$0'$}\newline
\indent \flip{$0'$}{$\done$}\newline
\dec.
\end{quote}

\begin{quote}
\tit{Other Prisoners' Algorithm:}\newline
\rep{$r$}\newline
\indent \flip{$0$}{$1$}\newline
\indent \flip{$0'$}{$1'$}.
\end{quote}

In the execution of this protocol, the leader selects a room in the $\off$ configuration, and changes it to the $0$ configuration. The prisoners then run the one-room protocol in that room (ignoring all of the other rooms, which are still in the $\off$ configuration). The leader then flips that room to the $0'$ configuration, and the prisoners run through the one-room protocol in that room \emph{again} to confirm that each prisoner has been in this room (using $0'$ and $1'$ instead of $0$ and $1$). After this, the leader puts that room in the $\done$ configuration and moves on to the next room. Note that the alternation between the $0/1$ version of the one room protocol and the $0'/1'$ version of the one room protocol is necessary here as otherwise the follower prisoners will not know when they have switched rooms.

The protocol just described requires at least $s=3$ switches per room.   We do not provide a full analysis of this protocol here, as in Section \ref{one room at a time sec} we discuss a refinement that gets by with only $s=2$.

\subsection{A Two-Switch Solution}\label{sec:2s}

The following is a relatively simple winning protocol for $s=2$ that works for arbitrary $n$ and $r$.  We name our four configurations (in some order) $0$ (the initial configuration), $1,\nextc,$ and $\ready$. The idea here is that instead of processing the rooms one at a time, we will process the \textit{prisoners} one at a time.

In our protocol, we have at most one ``active'' prisoner at a time; this prisoner will verify that they have visited every room by first flipping all rooms from the $0$ configuration to the $1$ configuration, and then flipping them all back. We then need a way to pass the torch to the next active prisoner---and ensure that the prisoner is finished being active is counted properly. In order to do this, the active prisoner will flip one room to the $\nextc$ configuration, and the designated leader will flip that room to the $\ready$ configuration (incrementing a counter in the process). The next prisoner who has not yet been active and sees the room in the $\ready$ configuration changes that room's configuration to $0$ and becomes the next active prisoner. We continue  this process until all prisoners have a chance to be active and visit all rooms---and then be counted.

Formally, the protocol is as follows:
\begin{quote}
\tit{Leader's Algorithm:}\newline
\rep{$r$} \com The Leader is the active prisoner. \newline
\indent \flip{$0$}{$1$} \newline
\rep{$r$}\newline
\indent \flip{$1$}{$0$} \newline
\flip{$0$}{$\nextc$} \com Sets a room to $\nextc$ to signal the next active prisoner. \newline
\rep{$n$} \com Counts the number of active prisoners (including himself). \newline
\indent \flip{$\nextc$}{$\ready$} \newline
\dec.
\end{quote}
\begin{quote}
\tit{Other Prisoners' Algorithm:}\newline
\flip{$\ready$}{$0$} \com Waiting for a room in $\ready$ state before becoming active.\newline
\rep{$r$}\newline
\indent \flip{$0$}{$1$} \com Visits all rooms.\newline
\rep{$r$}\newline
\indent \flip{$1$}{$0$} \com Resets all rooms.\newline
\flip{$0$}{$\nextc$}. \com Signals the leader that they are done.
\end{quote}

In order to analyze this protocol, we introduce some terminology.  We say that a prisoner is \emph{exhausted} if he is either a non-leader who has reached the end of his algorithm, or is the leader and has completed the first five lines of his algorithm.  We define an \emph{active} prisoner to be one who is either a non-leader who has completed the first line of his algorithm but is {not exhausted}, or a leader who is {not exhausted}.  We define a prisoner to be \emph{waiting} if he is neither active nor exhausted.  It is clear that each prisoner progresses sequentially from waiting to active to exhausted (expect for the leader, who is never waiting).

The correctness of our protocol depends heavily on the following invariant. At all times exactly one of the following holds:
\begin{itemize}
\item all rooms are in either the $0$ or $1$ configuration, and there is exactly one active prisoner; or
\item all rooms are in the $0$ configuration, except for a single room in the $\nextc$ or $\ready$ configuration, and there is no active prisoner.
\end{itemize}
We show that our invariant holds by induction. It is easy to check that the first condition holds in the initial configuration.  Now, when a prisoner becomes active, all rooms are in or are changed to the $0$ configuration.  As this prisoner remains active, no other prisoner will alter room configuration because non-active prisoners ignore rooms in state $0$ or $1$. Therefore, while active, this prisoner will change all rooms to the $1$ configuration and then change all rooms back to $0$ before becoming inactive. As this prisoner becomes inactive, they set one room to the $\nextc$ state, maintaining our invariant. This invariant continues to hold when the leader reconfigures this room from $\nextc$ to $\ready$ (and this is the only reconfiguration that can be performed by any of the inactive prisoners). This new state holds until the next prisoner becomes active.

In order to show that our protocol never declares incorrectly, we observe two more properties of it.  The first is that exhausted prisoners have visited all rooms; this follows from the preceding analysis and the fact that exhausted prisoners must have once been active.  Second, we claim that at the end of the $k$-{th} iteration of the final repeat loop on of the leader's algorithm, there are exactly $k$ exhausted prisoners.  We prove this by noting that our protocol cycles through the following three stages:
\begin{enumerate}
\item There is an active prisoner.
\item There is a room in the $\nextc$ configuration.
\item There is a room in the $\ready$ configuration.
\end{enumerate}
The claim follows from the fact that we increment the number of exhausted prisoners exactly when we transition from stage 1 to stage 2, and that we increment the counter on the repeat loop exactly when we transition from stage 2 to stage 3. Together, our claims imply that a declaration is made only when all $n$ prisoners are exhausted---and thus only when each prisoner has visited every room.  Thus, the protocol never declares incorrectly.

To prove that the protocol always terminates, we note that it always eventually either progresses to the next stage or the leader declares.  Since we can only transition from stage 2 to 3 a total of $n$ times, this proves that the protocol will eventually declare.  To show that the protocol will always progress from stage 1, we observe the following.  As a prisoner becomes active, all room are in the $0$ configuration. Since no other prisoner will alter any configurations during this stage, the active prisoner will switch every room to the $1$ configuration as he visits that room.  He will then switch each room to the $0$ configuration as he visits it.  He will then switch the next room he visits to the $\nextc$ configuration and move the algorithm to stage 2.  Stage 2 will always progress to stage 3 when the leader finds the room in the $\nextc$ configuration.  Stage 3 will progress to stage 1 when any waiting prisoner reaches the room in the $\ready$ configuration.  This will always happen eventually---unless there are no waiting prisoners, which only happens when all prisoners are exhausted, at which point the leader has reached the last line of his algorithm and is ready to declare.

So to summarize:
\begin{theorem}
There exists a winning protocol for the prisoners if there are two lightswitches per room and all switches start in the $\off$ configuration.
\end{theorem}

\begin{remark} Note that the protocol presented here satisfies the hypothesis of Lemma \ref{reductionlemma}, as no prisoner other than the leader will change the configuration of any room until he has seen a room in the $\ready$ configuration.  Thus with only $s=2$ switches in each room, we have a winning protocol for an arbitrary known starting configuration.
\end{remark}

\subsection{Solving the Problem One Room at a Time}\label{one room at a time sec}

As we already noted, the solution we just presented is substantially different from our earlier three-switch solution. Indeed, whereas our two-switch protocol proceeded one prisoner at a time, our three-switch protocol solved the problem one room at a time, with a single active room which is changing configuration, and the prisoners ensuring that each of them has visited that room before moving on to count visits to the next one. We might ask whether the one-room-at-a-time protocol can be made to work with only two switches---and in fact with some added complexity, we show that it can be.

To begin, we provide different names for the configurations.  We rename the four configurations $0$ (the initial configuration), $1$, ${\upc}$ and $\done$.

At a high level, in our protocol, the leader will select rooms one at a time and plays a game toggling the chosen room's state between two possibilities with the other prisoners (similar to the one-room, one-switch solution), before putting that room in the $\done$ configuration. It is easy to see how this works for the first room. The leader puts that room in the $\upc$ configuration and each other prisoner flips $\upc$ to $1$ once, while the leader flips it back from $1$ to $\upc$ $n-1$ times. At the end of this sequence of reconfigurations, the leader flips the room from $\upc$ to $\done$, establishing that all prisoners have visited the first room.

Naively, the leader could take another room from the $0$ configuration and change it to the $\upc$ configuration and he and the other prisoners could play the same game again. Unfortunately, this does not work so easily. The problem is that the non-leaders will not be able to distinguish between the second room being in the $\upc$ configuration and the \textit{first} room being in the $\upc$ configuration. So the leader needs a way to signal that the first round is over.

To do this, we note that in the first round there is never simultaneously a room in the $\upc$ configuration \textit{and} a room in the $1$ configuration. So if a prisoner seems a room in $\upc$ and then another in $1$, it must have been the case that they saw the active room twice, with that room reconfigured in the interim. However, the active room is only reconfigured a limited number of times. Therefore, no prisoner will ever see---during the first round---a long sequence of a room in the $\upc$ configuration followed by a room in the $1$ configuration followed by a room in the $\upc$ configuration and so on. This provides the leader a way to signal to the other prisoners that they have reached the second round. The leader does this by flipping all non-done rooms from the $0$ configuration to the $1$ configuration, and then flipping one of them to the $\upc$ configuration. The other prisoners will then eventually see a long sequence of alternative $\upc$ and $1$ configurations, and thus know that the second round has started.

Unfortunately, this idea does not work if the prisoners are toggling between $\upc$ and $1$ in the second round, as then the leader will see \emph{many} rooms in the $1$ configuration, which will prevent him from working with just a single room. This is solved by letting the prisoners toggle between $\upc$ and $0$ in the second round instead of $\upc$ and $1$---and each round after that, they must alternate between the two.

A final slight complication is that this signaling procedure does not work for the last round. This is because there is only one room left and so the prisoners will not be able to see many alternating configurations between $\upc$ and $1$, as there is only one non-$\done$ room. However, it turns out that we will not actually need the last round of the algorithm, as the signaling stage in the \textit{previous} round forces the prisoners to  visit both of the last two rooms in order to see the appropriate alternating sequence.

We present the algorithms for the leader and the others in the case where $r$ is odd.  The case where $r$ is even can be handled with a slight modification.

\begin{quote}
\tit{Leader's Algorithm:}\newline
\rep{$(r-1)/2$}\newline
\indent \flip{$0$}{$\upc$} \com{Readying next active room, start of the $0$-phase}\newline
\indent \rep{$n-2$} \newline
\indent \indent \flip{$1$}{$\upc$} \com{Counting other prisoners} \newline
\indent \flip{$1$}{$\done$} \com{Marking active room as done}\newline
\indent \rep{$r$ minus the number of rooms leader has configured to $\done$} \com{Transition phase}\newline
\indent \indent \flip{$0$}{$1$} \com{Ready rooms for the next phase}\newline
\indent \flip{$1$}{$\upc$} \com{Readying next active room, start of the $1$-phase}\newline
\indent \rep{$n-2$}\newline
\indent \indent \flip{$0$}{$\upc$} \com{Counting other prisoners}\newline
\indent \flip{$0$}{$\done$} \com{Marking active room as done}\newline
\indent \rep{$r$ minus the number of rooms leader has configured to $\done$} \com{Transition phase}\newline
\indent \indent \flip{$1$}{$0$}  \com{Ready rooms for the next phase}\newline
\dec
\end{quote}

\begin{quote}
\tit{Non-Leader's Algorithm:}\newline
\rep{$(r-1)/2$}\newline
\indent \rep{$n$} \com{Verify that rooms are set up for the $0$-phase} \newline
\indent \indent \seeC{$0$} \newline
\indent \indent \seeC{$\upc$} \newline
\indent \flip{$\upc$}{$1$} \com{Indicate visit of active room}\newline
\indent \rep{$n$} \com{Verify that rooms are set up for the $1$-phase}\newline
\indent \indent \seeC{$1$}\newline
\indent \indent \seeC{$\upc$}\newline
\indent \flip{$\upc$}{$0$} \com{Indicate visit of active room}\newline
\end{quote}
Note the SEE commands above:  These make sure that each prisoner stays in step with all of the others.  Waiting to see $0$s before flipping guarantees that they are in phase where prisoners are toggling between $\upc$ and $1$.  The number of repeats is necessary to ensure that they are not just seeing alternations between $0$ and $\upc$ in the active room.

To show that this protocol works we need some definitions.  During some parts of the protocol, the leader is flipping rooms between $0$ and $1$.  We call these times \emph{transition phases}.  When not in a transition phase, some rooms are in the $\done$ configuration and are called \emph{finished}.  Otherwise, either all but one of the unfinished rooms are in the $0$ configuration or all but one of the unfinished rooms are in the $1$ configuration. We call these periods the $0$-phase and $1$-phase respectively, and they correspond to the sections of the leader's algorithm where they are running \flip{$1$}{$\upc$} and \flip{$0$}{$\upc$} respectively (as indicated). During one of these phases there is one unfinished room, which we call the \emph{active room} which is toggled between $\upc$ and $1$ in the $0$-phase or between $\upc$ and $0$ in the $1$-phase. The remaining rooms, are not reconfigured at all during this phase.

We have left to prove that this description holds and that the protocol works. In the following analysis, we define the phases (transition, $0$, and $1$) based on where the leader is in their protocol as indicated in the comments above. In particular, we will need to prove the following:
\begin{enumerate}
\item During a $0$-phase, all rooms are in the $0$ or $\done$ configuration except for a single active room in the $1$ or $\upc$ configuration. Likewise, during a $1$-phase all rooms are in the $1$ or $\done$ configurations except for a single active room in the $0$ or $\upc$ configuration. Furthermore, at the start of this phase, the active room is in the $\upc$ configuration.
\item During the $0$-phase, no non-leader is at the \flip{$\upc$}{$0$} line of their protocol, and during the $1$-phase no non-leader is in the \flip{$\upc$}{$1$} line.
\item During the $0$-phase, each non-leader will flip the active room from $\upc$ to $1$ exactly once and will reconfigure no other rooms. During the $1$-phase, each non-leader will flip the active room from $\upc$ to $0$ exactly once and will reconfigure no other rooms.
\item During the transition phase, every room that has ever been an active room is in the $\done$ configuration. And the leader reconfigures all other rooms from $0$ to $1$ or from $1$ to $0$ while no other room reconfigurations take place.
\end{enumerate}
We show that these invariants hold by induction. We note that the statements about the $0$-phase and $1$-phase are symmetric, so we will only prove the former and under the assumption that these invariants hold for all previous phases.

We begin by showing that at the start of the $0$-phase all rooms are in the $\done$ or $0$ configuration with one in the $\upc$ configuration. This clearly holds after the first line or the leader's algorithm. Otherwise, invariant 4 implies that the leader reconfigured all non-$\done$ rooms to $0$ in the previous transition phase and reconfigured one of the $0$'s to $\upc$ at the start of the phase. We also note that at the start of the phase, each non-leader is between their \flip{$\upc$}{$0$} command and their \flip{$\upc$}{$1$} command. This is true at the start of the algorithm, and on later iterations, by assumption they executed their \flip{$\upc$}{$0$} in the last $1$-phase and have not executed \flip{$\upc$}{$1$} since.

From here we claim that during the $0$-phase, no room other than the active room is reconfigured. This is because reconfiguring a different room would require reconfiguring a room not in the $1$ or $\upc$ configuration. The leader does not do this until the transition phase. The non-leaders will not do this until they have seen a $1$ followed by an $\upc$ at least $n$ times. We claim that no non-leader sees this during the $0$-phase. This is because the first non-leader to see this must see the active room in these configurations (as no other room is in either the $1$ or $\upc$ configuration during this period). This in turn would imply that the leader must have reconfigured it from $1$ to $\upc$ at least $n$ times (since no non-leader is reconfiguring in this direction). However the leader reconfigures in this way at most $n-1$ times during this phase.

Next we will show that during the $0$-phase the active room will be reconfigured between the $1$ configuration and $\upc$ configuration $n-1$ times. We know that the leader will not progress with their protocol until they have reconfigured it from $1$ to $\upc$ a total of $n-1$ times. Furthermore, each of the $n-1$ non-leaders will have an opportunity to reconfigure the active room from $\upc$ to $1$ once during this phase. We claim that if the active room has not been reconfigured between $\upc$ and $1$ the full $n-1$ times, that it will eventually (assuming that each prisoner is lead into each room enough times) be reconfigured more. If the active room is currently in the $1$ configuration, the leader will eventually see it there and reconfigure it. If the active room is currently in the $\upc$ configuration and has been flipped from $\upc$ to $1$ fewer than $n-1$ times, there is at least one non-leader who has not reconfigured this room during this phase. This prisoner may still have some \seeC{$0$} and \seeC{$\upc$} commands to execute before their \flip{$\upc$}{$1$} command. However, if no other prisoner reconfigures the active room in the interim, they will eventually see the active room in the $\upc$ configuration followed by one of the unfinished rooms in the $0$ configuration enough times to finish their SEE' commands. Their next visit to the active room will cause it to be reconfigured.

The above implies that the protocol will eventually progress from the $0$-phase to the next transition phase. We note that it also implies that every prisoner visits the active room before this transition. This is because the leader must have reconfigured the active room from $1$ to $\upc$ a total of $n-1$ times. This is only possible if it was reconfigured from $\upc$ to $1$ this many times. However, each non-leader can only do so once. Therefore, by the end of the phase, each non-leader must have reconfigured the active room from $\upc$ to $1$.

We now discuss the transition phases. We consider the transition phase after a $0$-phase as the transition phase after a $1$-phase will by symmetric. At the start of the transition phase, the leader has just reconfigured the previously-active room to the $\done$ configuration, and all other rooms are in the $\done$ or $0$ configurations.

We next show that no non-leader reconfigures any room during this transition phase. This is because a non-leader will only reconfigure rooms found in the $\upc$ configuration. However, during the transition phase, no room is in the $\upc$ configuration, nor does the leader reconfigure any room into the $\upc$ configuration. During this transition phase, the leader does reconfigure a number of $0$ rooms to $1$ equal to the number of rooms in the $0$ configuration at the start of the phase. This is because at the start of the phase every room is in the $0$ or $\done$ configuration, so this number should be $r$ minus the number of rooms in the $\done$ configuration. However, since only the leader reconfigures rooms into the $\done$ configuration and since no prisoner reconfigures rooms out of the $\done$ configuration, the number of iterations in the leader's REPEAT loop is the number of rooms in the $0$ configuration. Since no rooms are being reconfigured by other prisoners, the leader will reconfigure each $0$ room to $1$ as they find it, and then reconfigure the next $1$-room to $\upc$, starting the next phase. We note that this leaves the rooms in the configurations needed at the start of the $1$-phase.

The above analysis shows that our invariants hold and that this protocol will eventually terminate. We have left to show that at the end of the protocol that every prisoner will have visited every room. Firstly, as we discuss above the active room in any $0$- or $1$-phase must be visited by every prisoner before progressing. Since the rooms in the $\done$ configuration are exactly the previously-active rooms, this means that every room in the $\done$-configuration was visited by every prisoner. We note that exactly $r-1$ rooms are put into the $\done$-configuration by the end of the protocol. This leaves a single remaining room to consider.

We note that this remaining room was the unique room in the $1$ configuration during the last $1$-phase. The leader must have visited this room because the leader must visit every non-finished room in every transition phase. To show that non-leaders visited this room, we note that each non-leader must have reconfigured the active room during the last $1$-phase. However, in order to do this, they must have seen $n$ alternations between rooms in the $1$- and $\upc$-configurations. But, as discussed above, they can have seen at most $n-1$ of these alternations in the previous $0$-phase and transition phase. Therefore, they must have seen this final room at least once during the last $1$-phase; this completes our argument.

\subsection{One Switch Does Not Suffice}

We now know that with two switches there are multiple strategies that allow the prisoners to win, for arbitrary $n$ and $r$.  We now show that one switch is \textit{insufficient} as long as $n\geq 2$ and $r\geq 5$.  

Given the sequence of rooms visited by prisoners and the actions which they take, we define the \emph{observed history} to be the ordered sequence of events describing a particular prisoner entering a room in some specified initial configuration and then leaving it in some specified configuration. For example, if the exercise starts with prisoner 1 visiting room 1 and changing the configuration from $\off$ to $\on$, and then prisoner 2 visiting and not changing the configuration, the observed history would look like this:
\begin{itemize}
\item Prisoner 1 enters a room in the $\off$ configuration and changes it to the $\on$ configuration.
\item Prisoner 2 enters a room in the $\on$ configuration and leaves it in the $\on$ configuration.
\end{itemize}

We say that a prisoner, $p$, \emph{owns} a room configuration, $c$, at some particular point in time if he has visited all rooms that are in configuration $c$ at that point in time.

We next say that a prisoner, $p$, \emph{provably owns} a room configuration, $c$, at some point in time if in all visit sequences with the same observed history, $p$ owns $c$ at that time.

\begin{lemma}\label{declareownLem}
A winning protocol for the prisoners will never declare unless all prisoners provably own all configurations.
\end{lemma}
\begin{proof}
Suppose that after some sequence of visits, the prisoners declare without prisoner $p$ provably owning configuration $c$.  This means that there is some sequence of visits with the same observed history in which $p$ does not own $c$.  Since the prisoners' behaviors (including their declarations and configuration changes) depend only on the observed history, this means that for that sequence of visits, the prisoners will declare before $p$ has visited all rooms in configuration $c$.  Thus the prisoners' strategy cannot be winning.
\end{proof}

{In order to reason about the concept of provable ownership we will need the following lemma:}
\begin{lemma}\label{ownchangeLem}
Whether or not a prisoner $p$ provably owns a configuration $c$ changes exactly in the following circumstances:
\begin{itemize}
\item Prisoner $p$ loses provable ownership of configuration $c$ when a room of a configuration not provably owned by $p$ is reconfigured to $c$ by some other prisoner.
\item Prisoner $p$ gains provable ownership of a configuration $c$ when he visits the only room in configuration $c$ or the only room in configuration $c$ is reconfigured to some other configuration.
\end{itemize}
\end{lemma}
\begin{proof}
Note that the number of rooms currently in each configuration can be inferred by the observed history. In particular, it can be determined when one of the situations in Lemma~\ref{ownchangeLem} has taken place by considering only the observed history.

Clearly $p$ can lose provable ownership of a configuration $c$ only if there is some possible sequence of visits with the same observed history in which he loses ownership of $c$.  This can happen only if some other prisoner reconfigures a room that $p$ has not visited into configuration $c$.  This in turn only happens when this other room is in a configuration, $c'$, which $p$ does not own. This in turn happens only if $p$ does not provably own $c'$.  Thus, $p$ can only lose provable ownership of $c$ if some other prisoner reconfigures a room from a configuration $c'$, not provably owned by $p$, to configuration $c$.  On the other hand, if $p$ does not provably own $c'$ and some other prisoner reconfigures $c'$ to $c$, there is some sequence of visits with the same observed history in which $p$ did not own $c'$ before this visit.  In this sequence $p$ had not visited all rooms currently in configuration $c'$, and thus without changing the observed history, we may have the last visit reconfigure a room that $p$ has not visited to configuration $c$. In this alternative visit sequence, we have the same observed history, but $p$ does not own $c$ at the end of it. Therefore, after such an event $p$ no longer provably owns $c$.

On the other hand, if $p$ visits the only room currently in configuration $c$, or if the only room currently in configuration $c$ is reconfigured into another configuration, it is clear that $p$ owns $c$ after this takes place. It is also easy to see that it is possible to determine when either of the above situations has taken place purely by considering the observed history. Therefore, under either of these situations, if $p$ did not previously provably own $c$, it gains such ownership. However, if $p$ did not provably own $c$ before and some visit caused to $p$ to gain such provable ownership. Then there must have previously been some sequence of visits with the same observed history in which $p$ did not own $c$ but for which any additional visit with the same observed data would cause $p$ to own $c$. This can happen only if the last room in configuration $c$ that $p$ had not yet visited is either visited by $p$ or reconfigured to another configuration. However, if more than one room was in configuration $c$ before this last visit, then the same observed history will be possible so that the last visit is not to the final room in configuration $c$ unvisited by $p$ (either the last visit is to a room in another configuration or it could be made to be to a different room in configuration $c$ without altering the observed history). Therefore, $p$ gains ownership of $c$ only if a visit is made to the unique room in configuration $c$ either by $p$ or by another prisoner who reconfigures it to a different configuration.
\end{proof}

Next we declare a prisoner \emph{finished} if under no circumstances will that prisoner ever again change the configuration of a room or declare.  We note the following lemma about when a prisoner may become finished.
\begin{lemma}\label{finishLem}
In a winning strategy for the prisoners, no prisoner may become finished before he provably owns all configurations at once.
\end{lemma}
\begin{proof}
Assume for sake of contradiction that there is a winning strategy that does not satisfy this property.
Suppose that there is some sequence of visits which causes prisoner $p$ to become finished while $p$ does not provably own all configurations.  This means that there is some sequence of visits with the same observed history for which $p$ does not own all the configurations, and thus has not visited all rooms.  Extend this sequence of visits arbitrarily until one of the prisoners declares (which they will do if each prisoner is led into each room sufficiently many times).  Then remove from this sequence all room visits that $p$ made since he became finished.  Since after this point, $p$ did not reconfigure any rooms, none of the other prisoners can distinguish these two visit sequences, and hence they will still declare.  On the other hand, in this new visit sequence $p$ will not have visited all rooms, and the prisoners must have declared incorrectly.
\end{proof}

We are now ready to prove our main result for this section:
\begin{theorem}
There is no winning strategy when $s=1$, $n\geq 2$, and $r\geq 5$.
\end{theorem}
The basic idea of our proof will be to construct, for any fixed strategy, a sequence $\sch$ of visits with the following properties:
\begin{itemize}
\item Until some prisoner becomes finished, no configuration with at least one room in that configuration is ever provably owned by more than one prisoner.
\item Until some prisoner becomes finished, no prisoner provably owns any configuration with more than two rooms in that configuration.
\item Each prisoner visits each room infinitely often.
\end{itemize}
We produce the desired sequence as follows.  We say that we \textit{extend our visit sequence directly} to mean that we execute the prisoner-room visit that has least recently occurred, with ties broken arbitrarily. We extend directly if:
\begin{enumerate}
\item Some prisoner is finished.
\item All rooms or all but one room are in the same configuration.
\item No prisoner provably owns any configuration.
\end{enumerate}
Otherwise:
\begin{enumerate}\setcounter{enumi}{3}
\item If there is no finished prisoner and there is some configuration $c$ with exactly $2$ rooms in configuration $c$, some prisoner $p$ who provably owns $c$ and no other configuration and no other prisoner who provably owns any configuration: In this case, let $p'$ be a prisoner other than $p$. Since $p'$ is not finished, there is some sequence of visits that will cause them to either reconfigure or a room or declare. In particular, there is some sequence of $0$'s and $1$'s so that if $p'$ is lead into rooms in those configurations in that order, then $p'$ will either reconfigure a room or declare at the end of that sequence. As there are currently rooms in both the $0$ and $1$ configuration, we can send $p'$ on such a sequence of visits, and we will do so.
\end{enumerate}
We note that assuming our invariants hold, the above list of possibilities is exhaustive. We have left to verify that this visit sequence satisfies the above invariants.  In other words we claim that at all times, one of the following conditions is true:
\begin{enumerate}
\item One of the configurations has at most one room in it, and no prisoner provably owns the other configuration.
\item There are at least two rooms in each configuration and no prisoner provably owns any configuration.
\item One configuration contains $2$ rooms. That configuration is provably owned by exactly one prisoner, but other than that, no prisoner provably owns any configuration.
\item Some prisoner is finished.
\end{enumerate}
We show by induction that at least one of these conditions will always hold.  If some prisoner is finished, he will always remain finished.

If Condition 1 currently holds, without loss of generality, there is at most one room in the $\on$ configuration.  Condition 1 will continue to hold until a second room is reconfigured into the $\on$ configuration, by some prisoner $p$.  At this time, by Lemma \ref{ownchangeLem} no prisoner provably owns any configuration, except for $p$, who might provably own the $\on$ configuration.  Hence, one of Conditions 2 or 3 are satisfied.

If Condition 2 is satisfied, a single visit cannot cause any prisoner to provably own any configuration, therefore after any visit either Condition 1 or Condition 2 is satisfied.

If Condition 3 is satisfied, any visit which does not reconfigure a room cannot change room ownership and so we will remain in condition 3. Otherwise, let $p,c$ be the unique pair of a prisoner who provably owns a configuration and let $p'$ be any other prisoner. If $p'$ reconfigures a room to configuration $c$, then no prisoner will provably own any configuration and we will be in Condition 2. If $p'$ reconfigures any room away from configuration $c$, we will be in Condition 1. Since our procedure ensures that only prisoners other than $p$ will reconfigure rooms in this case, our invariants are maintained.

Note that if we proceed directly infinitely often, each prisoner will visit each room infinitely often.  This will happen under our procedure since after applying case 4 in our sequence-generating procedure, we are left in one of the other cases, which will cause us to proceed directly again.

We note that if we run this sequence no prisoner can provably own all configurations until after some other prisoner becomes finished. This means that any strategy by the prisoners either has some prisoner become finished before provably owning all of the configuration (which implies that their strategy is not winning by Lemma \ref{finishLem}), or never leads to any prisoner provably owning all configurations. In the latter case, either the prisoners eventually declare (in which case their strategy is not winning by Lemma \ref{declareownLem}), or they never do. In that last case, the prisoners never declare despite each of them being lead into each room infinitely often, and so their strategy is not winning.

Thus, in any case, the prisoners' strategy cannot be winning.

\begin{remark} We note that with some additional complications, that it is possible to prove a similar result for as few as three rooms.  It is clear from the solution to the classic puzzle that that one switch suffices for a single room.  Whether or not one switch  suffices for two rooms is unclear.
\end{remark}

\section{Dimmer Switches}

Throughout the arguments presented above it has been useful to name our possible room configurations.  The observant reader will notice that when given $s$ switches, our real constraint is that we have only $2^s$ possible configurations; thus, for example in the proof that two switches suffice, a total of four names were used.  As a generalization of this use of multiple switches we can instead think of rooms as having a single ``dimmer switch'' with a number of possible configurations.  The arguments so far show that if the dimmer switch has four or more configurations, then the prisoners have a winning strategy---and if the switch has two or fewer configurations, the  prisoners do not (assuming that there are sufficiently many rooms and prisoners).

Whether or not the  prisoners have a winning strategy with three configurations is an open question.  However, we have found some strategies that seemingly come close.  For instance, using only three configurations, it is possible to guarantee that each prisoner will eventually know that he has visited all rooms.

For this protocol, we will call our room configurations $\on$, $\off$ and $\nextc$.  All of the prisoners  use the following algorithm:

\begin{quote}
\flip{$\nextc$}{$\off$}\newline
\rep{$r$}\newline
\indent \flip{$\off$}{$\on$} \newline
\rep{$r$}\newline
\indent \flip{$\on$}{$\off$} \newline
\flip{$\off$}{$\nextc$}\newline
\end{quote}
Furthermore, one of the prisoners, who will we call the leader, prepends the following command to his algorithm:
\begin{quote}
\flip{$\off$}{$\nextc$}\newline
\end{quote}

In the execution of this protocol, the leader will set one room to the $\nextc$ configuration.  Then one at a time, the prisoners will see a room in the $\nextc$ configuration, and change it to $\off$.  This prisoner will then reconfigure all rooms to $\on$ and then to $\off$ again before changing one room to the $\nextc$ configuration and letting the next prisoner have a chance.  Once each prisoner reaches the end of his algorithm, he can conclude that he has visited every room. Unfortunately, no prisoner is able to tell whether the other prisoners have visited all the rooms yet.

We note that the protocol just described is essentially our one-prisoner-at-a-time solution---but without the $\upc$ configuration, the leader has no way of counting the number of other prisoners who have finished.

\subsection{A Probability-$1$ Solution With 3 Configurations}\label{prob1sec}

While we do not know of a winning protocol for the case of a three-configuration switch, we have found a protocol almost as good.  The following protocol \emph{wins with probability 1}, by which we mean that
\begin{enumerate}
\item The prisoners will never declare incorrectly.
\item After any sequence of visits there is always some possible sequence of future visits of bounded length after which the prisoners will declare.
\end{enumerate}
In order to describe this protocol we need to add another command to our protocol language.

\begin{quote}
\osc{$c_1$}{$c_2$}{$k$}: Upon entering a room in configuration $c_1$, reconfigure it to configuration $c_2$. Upon entering a room in configuration $c_2$, reconfigure it to configuration $c_1$.  Continue this behavior until you have performed the former operation more times than you have performed the latter operation.
\end{quote}

For this protocol, we label the prisoners  $1,2,\ldots,n$, and label the room configurations label  $0,1,2$, with $0$ as the starting configuration.  Prisoner $k$'s algorithm will be as follows for $k\neq 1,n$:

\begin{quote}
\tit{Prisoner $k$'s Algorithm ($k\neq 1,n$):}\newline
\rep{$k-1$}\newline
\indent \flip{$1$}{$0$}\newline
\indent \flip{$2$}{$1$}\newline
\indent \flip{$0$}{$2$}\newline
\seeC{$1$} \com{Transitioning}\newline
\rep{$n+r-1$} \com{Active, $0$-phase}\newline
\indent \flip{$0$}{$1$}\newline
\rep{$n+r-1$} \com{$1$-phase}\newline
\indent \flip{$1$}{$2$}\newline
\rep{$n+r-1$} \com{$2$-phase}\newline
\indent \flip{$2$}{$0$}\newline
\osc{$1$}{$0$}{$1$} \com{Transitioning}\newline
\rep{$n-k$} \com{No longer active}\newline
\indent \flip{$2$}{$1$}\newline
\indent \flip{$0$}{$2$}\newline
\indent \flip{$1$}{$0$}.\newline
\end{quote}

The algorithms for prisoners $1$ and $n$ are similar.  To get prisoner $1$'s algorithm, we remove the initial REPEAT block, and the initial SEE command.  To get prisoner $n$'s algorithm, we replace the OSCILLATE command by a DECLARE command, and remove the succeeding REPEAT block:

\begin{quote}
\tit{Prisoner $1$'s Algorithm:}\newline
\rep{$n+r-1$}\newline
\indent \flip{$0$}{$1$}\newline
\rep{$n+r-1$}\newline
\indent \flip{$1$}{$2$}\newline
\rep{$n+r-1$}\newline
\indent \flip{$2$}{$0$}\newline
\osc{$1$}{$0$}{$1$}\newline
\rep{$n-1$}\newline
\indent \flip{$2$}{$1$}\newline
\indent \flip{$0$}{$2$}\newline
\indent \flip{$1$}{$0$};\newline
\end{quote}

\begin{quote}
\tit{Prisoner $n$'s Algorithm:}\newline
\rep{$n-1$}\newline
\indent \flip{$1$}{$0$}\newline
\indent \flip{$2$}{$1$}\newline
\indent \flip{$0$}{$2$}\newline
\seeC{$1$}\newline
\rep{$n+r-1$}\newline
\indent \flip{$0$}{$1$}\newline
\rep{$n+r-1$}\newline
\indent \flip{$1$}{$2$}\newline
\rep{$n+r-1$}\newline
\indent \flip{$2$}{$0$}.\newline
\dec \newline
\end{quote}

At a high level the execution of the protocol will work as follows. Each prisoner one at a time becomes \emph{active} (when they are between the SEE command and OSCILLATE command in their execution). The active prisoner will turns all the $0$s to $1$s, then all the $1$s to $2$s then all the $2$s back to $0$s. Meanwhile the other prisoners will resist this change by flipping rooms in the opposite direction once per prisoner per step. There are two things worth noting about this. Firstly, it guarantees that the active prisoner visits every room because the number of rooms flipped from $0$ to $1$ (or from $1$ to $2$ or from $2$ to $0$) is equal to the number of rooms plus the number of other prisoners flipping them in the other directions. Secondly, the other prisoners' resistance allows them to keep track of where in the algorithm they are. To see this, note that while the active prisoner is reconfiguring $0$s to $1$s, another prisoner might reconfigure a $1$ back to a $0$, but they will not be able to execute their next command (flipping a $2$ to a $1$) until the active prisoner moves to their next phase (flipping $1$s to $2$s).

The one difficulty with this idea is how we switch from one active prisoner to the next. The issue is that the new active prisoner needs to wait until the previous one is finished turning $2$'s into $0$ before they start turning $0$s into $1$s. We note that if given access to a fourth configuration, $\upc$, we could have the previous active prisoner reconfigure a room to $\upc$ when they are done, signaling to the new one that they are ready. This would give an algorithm similar to the one-prisoner at a time algorithm. Otherwise, a simple way to signal that they are ready is to flip a room into the $1$ configuration. This would work except that the other prisoners are reconfiguring $1$s to $0$s and might destroy the signal before the new active prisoner sees it. This could be fixed if the old active prisoner reconfigured $r-2$ rooms from $0$ to $1$, however, this introduces a new problem. In particular, with some other prisoners reconfiguring $1$s to $0$s and some reconfiguring $0$s to $1$s, the active prisoner will not be able to tell whether or not they are all finished, since if one $0$ to $1$ was skipped and one $1$ to $0$ was skipped, there would be no way to know. In order to fix this, we want to instead guarantee that the old active prisoner on net turns more $1$s to $0$s than $0$s to $1$s. However, he cannot do this immediately as he might simply flip many $0$s to $1$s and then flip them back without the new active prisoner seeing the signal. We fix this with the oscillate command. This ensures that the old active prisoner keeps reconfiguring rooms back and forth between $0$ and $1$ until somebody (who must in this case be the new active prisoner) starts configuring $0$s to $1$s.

To make things rigorous, we introduce some notation. A prisoner executing their SEE or OSCILLATE commands is called \emph{transitioning} a prisoner between those commands in their execution is called \emph{active}. If they are in their repeat loop, they are in the $0$- $1$- or $2$-phase as noted above. We make the following claims about the execution of algorithm:
\begin{enumerate}
\item There is never more than one active prisoner at a time.
\item The prisoners become active in order.
\item During a $0$-phase, or while there is no active prisoner all rooms are in configuration $0$ or $1$.
\item During a $1$-phase all rooms are in configuration $1$ or $2$.
\item During a $2$-phase all rooms are in configuration $2$ or $0$.
\item During another prisoner's $0$-phase, each other prisoner is on one of their \flip{$1$}{$0$},\flip{$2$}{$1$}, or \osc{$1$}{$0$}{$1$} commands.
\item During another prisoner's $1$-phase, each other prisoner is on one of their \flip{$2$}{$1$} or \flip{$0$}{$2$} commands.
\item During another prisoner's $2$-phase, each other prisoner is on one of their \flip{$0$}{$2$} or \flip{$1$}{$0$} commands.
\item At the start of the protocol and when a prisoner first switches from active to transitioning, all rooms are in the $0$-configuration and all other prisoners are executing a \flip{$1$}{$0$} or \seeC{$1$} command, with at most one prisoner in the latter state.
\item At the start of a $1$-phase, all rooms are in the $1$-configuration, and all non-active prisoners are executing a \flip{$2$}{$1$} command.
\item At the start of a $2$-phase, all rooms are in the $2$-configuration, and all non-active prisoners are executing a \flip{$0$}{$2$} command.
\end{enumerate}

To show that these continue to hold, we assume that they do at the end of a given phase, and show that they still do at the end of the next phase. The analysis for a $1$-phase is easy. At the start of a $1$-phase all rooms are in the $1$ configuration and all non-active prisoners are executing \flip{$2$}{$1$} and the active prisoner is executing their repeat loop of \flip{$1$}{$2$}. It is clear that no non-active prisoners will be able to execute their next command (\flip{$0$}{$2$}) until the active prisoner has moved on to the next phase. The active prisoner will not be able to do this until they have flipped $n+r-1$ rooms from $1$ to $2$. As there are only $r$ rooms available, they cannot do this unless a total of $n-1$ rooms (with multiplicity) are flipped from $2$ back to $1$. This can only happen if each other prisoner completes their \flip{$2$}{$1$} command. At the end of this, all rooms will have been changed to state $2$ and all other prisoners will be on their \flip{$0$}{$2$} commands showing that the state at the start of the next phase is as desired. We also note that this implies that the active prisoner visits every room before the end of this phase.

The analysis for phase-$2$ is similar. The one difference is that we note that exactly one prisoner ends with a \seeC{$1$} command rather than a \flip{$1$}{$0$} command. This is because each pre-transitioning prisoner executes exactly one command per phase. Therefore if the $k^{th}$ prisoner just finished being active, exactly the $(k+1)^{st}$ prisoner is on their \seeC{$1$} command.

The analysis for phase-$0$ (actually starting from where the previous active prisoner switched to being transitioning) is slightly more complicated. One prisoner started at their \seeC{$1$} command, we will call them active, although technically they are transitioning until they see a room in configuration $1$. We note that after seeing this room, they will try to flip $n+r-1$ rooms from $0$ to $1$ before moving on to the next phase. Meanwhile, the other prisoners are either executing \flip{$1$}{$0$} or \osc{$1$}{$0$}{$1$}. We note that after that command, the other prisoner will try to execute \flip{$2$}{$1$}, but will be unable to as no room will be in the $2$ configuration until the next phase. The active prisoner needs to flip $n+r-1$ rooms from $0$ to $1$. There are a total of $r$ rooms, initially in the $0$-configuration. It will be possible to reconfigure rooms into the $1$ configuration $n+r-1$ times only if the other prisoners in aggregate reconfigure rooms from $1$ to $0$ at least $n-1$ times more often than they reconfigure rooms from $0$ to $1$. We note that each other prisoner may do so on net at most $1$ time. Therefore, we can only transition to the next phase once all non-active prisoners have executed their \flip{$1$}{$0$} or \osc{$1$}{$0$}{$1$} commands (but not their next \flip{$2$}{$1$} command) and all rooms have been reconfigured to $1$.

From the above analysis, we note that each non-active prisoner executes exactly one command per phase, and that each active prisoner visits all rooms before becoming non-active again. From this it is easy to see that if the $n^{th}$ prisoner declares, that every prisoner must have been active at some point, and therefore every prisoner must have visited every room at least once.

We have left to show that this happens with probability $1$. For this we will show that from any reachable state, there is always some continuation that causes the protocol to progress to the next phase. For example, starting at the beginning of a $1$-phase, at any point until the next phase, the number of rooms in the $1$ configuration plus the number of non-active prisoners on their \flip{$2$}{$1$} commands is always the number of iterations left on the active prisoner's loop. This means that there is always either a $2$-room to visit for one of the non-active prisoners still on their \flip{$2$}{$1$} command or a room in the $1$-configuration for the active -prisoner to visit. Therefore, if each prisoner visits each room infinitely often, eventually the non-active prisoners will complete their \flip{$2$}{$1$} commands and the leader will complete their loop and the phase will end.

The analysis for the $2$-phase is analogous. The analysis starting after the end of the $2$-phase is slightly more complicated. Firstly, we show that there is always a way for the next active prisoner to execute their \seeC{$1$} command. This is because until they do, the previously active prisoner will be executing their \osc{$1$}{$0$}{$1$} command. This in turn is because they cannot end until they have reconfigured more rooms from $1$ to $0$ than from $0$ to $1$, but until the next prisoner becomes active, no other prisoner is reconfiguring $0$ to $1$. While the previous active prisoner is oscillating, there will always be the possibility that they reconfigure a room to $1$, which is then seen by the next active prisoner. Once this has happened, we claim that (until the end of the phase) it will always be possible for the active prisoner to reconfigure a $0$ to a $1$ or a non-active prisoner to complete their current command. Since these can only happen a bounded number of times during the phase, there will always be a way to proceed to the next phase. To show this, if the oscillate command has not completed, it will always be possible for that prisoner to reconfigure some room to a $0$ so that the active prisoner can later reconfigure it to a $1$. If the oscillate command has completed, it is easy to see that the number of rooms in the $0$ configuration plus the number of prisoners who have not completed their \flip{$1$}{$0$} command is the number of remaining iterations in the active prisoner's repeat loop. From this it is easy to see that there is always either a room in the $0$ configuration for the active prisoner to flip to $1$, or a room in the $1$ configuration for some non-active prisoner (still on their \flip{$1$}{$0$} command) to reconfigure to $0$. This shows that it is always possible to make progress, completing our argument.

\section{A Probability-$\epsilon$ Solution With $2$ Configurations}

In the last section, we found a probability-$1$ algorithm for three configurations.  Unfortunately our impossibility proof for $2$ configurations does not generalize to algorithms merely working with probability $1$.  Although we do not have a two-configuration protocol that succeeds with probability $1$, we do have an algorithm that satisfies another interesting condition.

We define an algorithm to \emph{win with probability $\epsilon$} if the prisoners never declare incorrectly and do declare in some sequence of visits.  The difference between probability $1$ and probability $\epsilon$ is that in the latter case it may be possible to become stuck.  Essentially, having a probability $\epsilon$ algorithm means that you have a way of proving that everyone has been to every room.  Given the two room configurations, $0$ and $1$, with $0$ the starting configuration, we produce the following algorithm, where the prisoners are $p_0,p_1,p_2,\ldots,p_{n-1}$ (intuitively in the execution that causes them to win, the prisoners act in order $p_0$ first then $p_1$ and so on):

\begin{quote}
\tit{$p_k$'s $(k\neq n-1)$ Algorithm:}\newline
\rep{$r+k$}\com{Startup phase}\newline
\indent \flip{$0$}{$1$} \com{Started after first command executed} \newline
\rep{$r$} \com{Check phase} \newline
\indent \flip{$1$}{$0$}\newline
\rep{$r$}\newline
\indent \flip{$0$}{$1$} \newline
\rep{$r+k+1$}\com{Cooldown phase} \newline
\indent \flip{$1$}{$0$} \com{Finished} \newline
\end{quote}

\begin{quote}
\tit{$p_{n-1}$'s Algorithm:}\newline
\rep{$r+n-1$}\newline
\indent \flip{$0$}{$1$} \newline
\rep{$r$}\newline
\indent \flip{$1$}{$0$} \newline
\rep{$r$}\newline
\indent \flip{$0$}{$1$} \newline
\dec
\end{quote}

We first define a few phases of the algorithm.  Once a prisoner has executed his first flip command we declare him to have \emph{started}.  While in the first REPEAT loop, we say that a prisoner is in the \emph{startup phase}.  During the next two REPEAT loops, we say that prisoner is in the \emph{check phase}.  While in the last loop, we say he is in the \emph{cooldown phase}---and when he finishes it, they are \emph{finished}.

Firstly, we note that there is some sequence of visits that cause the prisoners to declare. The required visit sequence is as follows. Firstly, $p_0$ visits each room in sequence (changing them all to $1$s and finishing their startup phase), then visits them all again (changing the rooms back to $0$ and finishing their second repeat loop), visiting all rooms a third time (finishing their check phase, and setting them to $1$), and then visiting a fourth time (setting all rooms to $0$ and finishing all but the last step in their cooldown phase). Then $p_1$ visits a room $R$ followed by $p_0$ visiting $R$. Prisoner $p_1$ flips $R$ to configuration $1$ and back with $p_0$ finishing their cooldown phase and $p_1$ executing the first command in their startup phase. Then $p_1$ visits each room in order four times. As before, this leaves all rooms in the $0$ configuration with $p_1$ having finished all but the last two steps of their cooldown phase. Next, we have $p_2$ and $p_1$ alternate visits to room $R$ twice. This finishes $p_1$'s cooldown and the first two steps of $p_2$'s startup. We continue in this manner.

In general, we will reach a state where all rooms are in the $0$ configuration, $p_0,\ldots,p_{k-1}$ have finished, $p_k$ has completed all but the last $k+1$ steps of their cooldown phase, and none of $p_{k+1},\ldots,p_n$ have started. We then have $p_k$ and $p_{k+1}$ alternate visits to room $R$ a total of $k+1$ times. This causes $p_k$ to finish their cooldown and for $p_{k+1}$ to complete the first $k+1$ steps of their startup. We then have $p_{k+1}$ visit all rooms in order four times, leaving them all in configuration $0$, with $p_{k+1}$ having completed all but the last $k+2$ steps of their cooldown (or declaring if $k+1=n-1$). This leaves us in the same situation as we started with but for $k+1$. Continuing in this manner, we will reach it for $p_{k+2},p_{k+3},\ldots,p_{n-1}$, at which point we will declare.

We now need to verify the more difficult assertion that this strategy only declares after all prisoners have visited all rooms.  Most of our argument will be based on one simple fact.  At any time during the execution of this algorithm, the difference in the total number of \flip{$0$}{$1$} commands executed by all prisoners combined and the number of \flip{$1$}{$0$} commands executed by all prisoners combined is between $0$ and $r$ inclusive (since this difference is the number of rooms currently in configuration $1$).  For each prisoner we define their \emph{imbalance} to be the difference in the number of \flip{$0$}{$1$}'s they have executed and the number of \flip{$1$}{$0$}'s they have executed.  Hence the sum of all prisoner's imbalance is between $0$ and $r$ inclusive.  We note the following easily verified facts about prisoners' imbalances:
\begin{enumerate}
\item After a prisoner starts, their imbalance is positive until their cooldown phase.
\item A prisoner's imbalance is non-negative until they are finished when it becomes -1.
\item At the end of a prisoner's startup phase and at the end of their check phase, their imbalance is $r+k$.
\end{enumerate}

Note that (2) and (3) above imply that $p_k$ cannot finish his startup phase until at least $k$ other prisoner's have finished.  Since a prisoner must end their startup phase before finishing, $p_0$ is the only prisoner than can finish before any other.  Similarly, $p_1$ is the only prisoner that can finish after only $p_0$ has.  Continuing with this logic, we conclude that if prisoners finish at all they must do so in the order $p_0,p_1,p_2,\ldots$. Additionally, $p_k$ cannot even finish their startup phase until $p_0,\ldots,p_{k-1}$ have finished. Note therefore, that when $p_k$ ends his startup phase or ends his check phase, only $p_0,\ldots,p_{k-1}$ can have finished and that all other prisoners must have non-negative imbalance. Since $p_k$ has imbalance $r+k$ and the total imbalance is at most $r$, this means that it must be the case that $p_0,\ldots,p_{k-1}$ have imbalance $-1$ (which implies that they have finished) and that $p_{k+1},\ldots,p_{n-1}$ must have imbalance $0$ (which means that they haven't started). This means that no other prisoners reconfigure any rooms during $p_k$'s check phase. In particular, it means that during the first loop of $p_k$'s check phase they must flip every room from $1$ to $0$ (and thus must visit every room). Therefore, every finished prisoner must have visited every room. Furthermore, when the prisoners declare, $p_{n-1}$ is finished. This implies that $p_0,\ldots,p_{n-1}$ must all have finished. And thus, every prisoner must have visited every room.

\begin{remark} We note that although the prisoners never declare incorrectly here, it is very easy for them to get stuck. The above proof shows that in order for them to declare, it must be the case that $p_k$ doesn't start until all of $p_0,\ldots,p_{k-2}$ have finished. Of course if $p_k$ (for some $k\geq 1$) is the first prisoner to visit any room, this will never happen. In particular, all of the rooms will be reconfigured to the $1$ configuration before any prisoner has finished their startup phase, and there will be no way to make further progress.
\end{remark}

\section{Corner Cases and Related Problems}

We close by discussing a number of special cases, in which sharper results can be obtained, along with some variants on our problem.

\subsection{$2$ Rooms, $3$ Configurations}

We have a three-configuration solution for the special case of $r=2$.   For this solution, we call our configurations $\upc,\on$, and $\off$---with $\off$ representing the initial configuration.  We have a single leader, whose algorithm is as follows.

\begin{quote}
\tit{Leader's Algorithm:}\newline
\flip{$\off$}{$\upc$}\newline
\rep{$n-1$}\newline
\indent \flip{$\on$}{$\off$}\newline
\dec\newline
\end{quote}

All other prisoners use the following algorithm.

\begin{quote}
\tit{Other Prisoners' Algorithm:}\newline
\seeC{$\upc$}\newline
\flip{$\off$}{$\on$}\newline
\end{quote}

Essentially, after the leader produces a single room in the $\upc$ state, the prisoners execute the standard one-room protocol in the other room, with the proviso that they do nothing until they have seen the room in the $\upc$ state.
\begin{remark} This idea also provides us with a somewhat silly protocol for $r=3$ with four configurations.\end{remark}

\subsection{Small $n$}

In Section \ref{naive solutions section}, we described a solution using $n+1$ configurations in which prisoner $k$ would change all the rooms from configuration $k-1$ to configuration $k$.  Although this is somewhat inefficient for $n\geq 3$, it provides new solutions when $n=1 $ or $n= 2$.

\subsection{Unknown Starting Configuration with Infinitely Many Room Configurations}

We note that the proof of Theorem \ref{unknownStartThrm} actually requires that each room has only a finite number of possible configurations---as it happens, this is actually necessary. In particular, as we show now, if there are infinitely many configurations, then there \emph{is} a strategy that works for arbitrary starting configurations.

For simplicity, we assume that there are countably infinitely many configurations and that these configurations correspond to finite-length alphanumeric strings, thought of as writing on the walls of the room. The prisoners' strategy here is actually fairly simple. Upon entering a room, each prisoner appends to that room's transcript their name followed by the number of rooms that prisoner has visited so far. We claim that with this simple strategy, eventually some prisoner will have enough information to be able to conclude that each prisoner has visited every room.

To start the analysis, we note first that eventually all $r$ rooms will become distinguishable from each other. In particular, if a prisoner ever sees $r$ rooms where for no pair of these rooms is the transcript of one a prefix of the transcript of the other, these rooms must all be distinct (as the transcript of a room can only be modified by appending new text). To show this, we consider some particular prisoner, Barry. Upon visiting his $k$-{th} room, he will append ``Barry$k$'' to the transcript of that room. Now some rooms may have strings of the form ``Barry$m$'' in their initial transcript, but since these initial transcripts are finite there is a maximum such value of $m$ that ever appears. Call $M$ the largest such value of $m$. Once Barry makes his $k$-{th} visit to any room for any $k>M$, his text ``Barry$k$'' in that room (followed up by another prisoner's name rather than more digits for the number) will never appear in any other room. Once he has made such visits to all $r$ rooms, the rooms will thereafter be distinguishable from each other.

So eventually, a prisoner will see rooms with transcripts $T_1,T_2,\ldots,T_r$ none of which is a prefix of any other. At some later point, this prisoner will see rooms with transcript $T_i$ followed by some list of names and numbers that include the names of every prisoner. At that point, it must be the case that every prisoner has visited the room that was in configuration $T_i$. Once any prisoner has seen this occur for all $i$ with $1\leq i \leq r$, they can safely declare.

\subsection{Symmetric Strategy}

One interesting modification to our problem would be the additional requirement that the prisoners use identical strategies.  Essentially none of our protocols satisfy this property.  In general this problem seems to be much harder (although if the prisoners are allowed to specify a starting configuration, they can start with a single room in a special configuration and structure the protocol so that the first person to see that configuration becomes designated leader).

\subsection{Repeated Entries}

A substantially easier modification to the rules  requires that each prisoner visit each room $\ell\geq 1$ times before the prisoners declare.  This problem is not significantly more difficult than the original, as several of our algorithms can be easily modified to accommodate it.  For example, our original two-switch solution can be modified so that each prisoner flips all rooms on and then all rooms off $\ell$ times.  Essentially all of the protocols presented in this paper have similar modifications.

\subsection{Multiple Declarations}

Another modification of the problem is obtained by requiring that all prisoners declare at some point after they have all visited every room.  This can be done with four-state switches using a slight modification of the protocol given in Section \ref{one room at a time sec}: the leader puts a room in the $\upc$ state to denote that it is time to declare.  In particular, we append
\begin{quote}
\flip{$\done$}{$\upc$}\newline
\end{quote}
to the leader's algorithm, and append
\begin{quote}
\seeC{$\upc$}\newline
\dec\newline
\end{quote}
to all other prisoners' algorithms.

\subsection{Forced Flipping of Switches}

Another modification would be the require that upon eeach visit to a room, a prisoner \emph{must} reconfigure that room's state in some way.  It is not clear that any of our existing protocols generalize to this alternate setting directly, but any protocol for the original problem can be extended to this case by doubling the number of room configurations.  This is done by replacing each configuration by a pair of new configurations, which are treated as equivalent for purposes of the protocol, except that a prisoner can toggle between them if no other configuration change is desired.

Put another way, we can accomplish this by adding an additional switch to each room. This switch will have no effect on the rest of our protocol save that any prisoner visiting a room will always flip that switch in addition to whatever else they were going to do.

\subsection{Limited Reconfiguration}

More generally, the warden could impose essentially arbitrary restrictions on which configurations can be reconfigured into which other configurations in a single visit.  We cannot say much about the problem in this level of generality, and leave it to prisoners craftier than us.

\bibliographystyle{amsplain}
\bibliography{prisoners}

\providecommand{\bysame}{\leavevmode\hbox to3em{\hrulefill}\thinspace}
\providecommand{\MR}{\relax\ifhmode\unskip\space\fi MR }
\providecommand{\MRhref}[2]{%
  \href{http://www.ams.org/mathscinet-getitem?mr=#1}{#2}
}
\providecommand{\href}[2]{#2}
\begin{thebibliography}{1}

\bibitem{Emiss}
Joe~P. Buhler and Elwyn~R. Berlekamp, \emph{Puzzle~4}, The Emissary \textbf{5}
  (2002), no.~2, 11.

\bibitem{dehaye2003one}
Paul-Olivier Dehaye, Daniel Ford, and Henry Segerman, \emph{One hundred
  prisoners and a lightbulb}, The Mathematical Intelligencer \textbf{25}
  (2003), no.~4, 53--61.

\bibitem{BOP}
Scott~Duke Kominers, \emph{Kominers's conundrums: The warden has a
  brainteaser}, Bloomberg Opinion (April 25, 2020).

\bibitem{HCMR}
Scott~Duke Kominers, Paul Kominers, and Justin Chen, \emph{Problem {S08-2}},
  The Harvard College Mathematics Review \textbf{2} (2008), no.~1, 93.

\bibitem{HCMR2}
\bysame, \emph{Problem {S08-2} (corrected)}, The Harvard College Mathematics
  Review \textbf{2} (2008), no.~2, 96.

\bibitem{NPR}
Car~Talk~Radio Show, \emph{Prison switcharoo}, National Public Radio, 2003.

\bibitem{Winkler}
Peter Winkler, \emph{Mathematical puzzles: A connoisseur's collection}, AK
  Peters, 2004.

\end{thebibliography}
\end{document}